\documentclass[11pt, letterpaper]{article} 

\newif\ifsoda

\newif\ifarxiv
\arxivtrue 

\ifsoda
\usepackage[letterpaper]{geometry}

\usepackage{ltexpprt}
\usepackage{hyperref}

\usepackage{amsmath, amssymb}
\usepackage{graphicx}
\usepackage{verbatim}

\newtheorem{observation}{Observation}

\fi

\ifarxiv
\usepackage[margin=1in]{geometry}
\usepackage{charter}

\usepackage{amsfonts, amsmath, amssymb, amsthm}
\usepackage{bm}
\usepackage{verbatim}
\usepackage{color}
\usepackage{euscript}
\usepackage{graphicx}
\usepackage[usenames,dvipsnames]{xcolor}

\usepackage{xspace}
\usepackage{wrapfig}

\newtheorem{theorem}{Theorem}
\newtheorem{lemma}{Lemma}

\newtheorem{observation}{Observation}

\fi

\begin{document}

\ifarxiv
\title{4D Range Reporting in the Pointer Machine Model in Almost-Optimal Time}

\date{}

\author{Yakov Nekrich\thanks{Department of Computer Science, Michigan Technological University.}
  \and Saladi Rahul\thanks{Department of Computer Science and Automation, Indian Institute of Science.}}

\maketitle
\fi

\ifarxiv
\begin{abstract}
\tolerance=400
In the orthogonal range reporting problem we must pre-process a set $P$ of multi-dimensional points, so that for any axis-parallel query rectangle $q$ all points from $q\cap P$ can be reported efficiently. In this paper we study the query complexity of multi-dimensional orthogonal range reporting in the pointer machine model.  We present a data structure that answers four-dimensional orthogonal range reporting queries in almost-optimal time $O(\log n\log\log n + k)$ and uses $O(n\log^4 n)$ space, where $n$ is the number of points in $P$ and $k$ is the number of points in $q\cap P$ . This is the first data structure with nearly-linear space usage that achieves almost-optimal query time in 4d.  This result can be immediately generalized to $d\ge 4$ dimensions: we show that there is a data structure  supporting $d$-dimensional range reporting queries in time $O(\log^{d-3} n\log\log n+k)$  for any constant $d\ge 4$.
\end{abstract}

\thispagestyle{empty}
\clearpage
\setcounter{page}{1}
\fi

\section{Introduction}
In the orthogonal range  reporting problem we must pre-process a set $P$ of multi-dimensional points, so that for any axis-parallel query rectangle $q$ all points from $q\cap P$ can be reported efficiently. Orthogonal range reporting was studied extensively in both RAM model and the pointer machine model, see e.g.,\cite{Chazelle86,Chazelle88,AlstrupBR00,Afshani08,AfshaniAL09,AfshaniAL10,AfshaniAL12,ChanLP11,Chan13,Nekrich20,Nekrich21soda}.   In this paper we investigate the complexity of this problem in the pointer machine model. We  present a data structure that answers four-dimensional queries in almost-optimal time $O(\log n\log\log n + k)$,  where $n$ is the number of points in $P$ and $k$ is the number of points in $q\cap P$.  Our result immediately extends to $d> 4$ dimensions: there exists a data structure supporting orthogonal range reporting queries in time
$O(\log^{d-3} n\log\log n + k)$ for any constant $d\ge 4$.

The pointer machine (PM) model was initially introduced by Tarjan~\cite{Tarjan79}.
Unlike RAM, in the pointer machine model each memory cell can be accessed through a series of pointers  only.
Informally, we can view the pointer machine as the computational model in which the use of arrays is not allowed. 
A number of important geometric problems was studied extensively in the  pointer machine model, see e.g.,\cite{ChazelleG86,Chazelle88,MehlhornNA88,c90,Chazelle90b,Afshani08,AfshaniAL09,AfshaniAL10,AfshaniAL12}. 
It is important to investigate the differences in  computational power of the PM and the standard RAM model. Understanding their respective  limits and advantages  can be potentially useful  for both models.

In the pointer machine model, two-dimensional range reporting queries can be answered in time $O(\log n + k)$~\cite{Bentley79,Chazelle86}. Three-dimensional orthogonal range reporting queries can be also answered in $O(\log n + k)$ time~\cite{Afshani08}.  This query time is optimal because predecessor queries can be reduced to orthogonal range reporting in 2d.  Thus, in the PM model  we completely understand the query complexity of range reporting in $d\le 3$ dimensions. However it is not known what is the optimal query time of orthogonal range reporting in $d> 3$ dimensions. Using range trees~\cite{Bentley79}, we can answer 4d range reporting queries in $O(\log^2 n + k)$ time. Afshani et al.~\cite{AfshaniAL09,AfshaniAL10} described a data structure that answers queries in $O((\log n/\log\log n)^2 + k)$ time. In a subsequent paper~\cite{AfshaniAL12} the same authors significantly improved the query time and obtained a data structure with $O(\log^{3/2}n+ k)$ query time.  See Table~\ref{tab:results}. However the optimal complexity of four-dimensional range reporting queries remains an intriguing question. Existence of a data structure that supports four-dimensional queries in  optimal $O(\log n + k)$ time was asked as an open question in several research papers and surveys, see e.g.,~\cite{AfshaniAL12} and in~\cite{Agarwal17survey}.

In this paper we come very close to answering this open question and show that four-dimensional range reporting can be answered in $O(\log n \log\log n + k)$ time in the PM model.  We describe a data structure that uses $O(n \log n)$ space and supports four-dimensional dominance\footnote{We refer to Section~\ref{sec:prelim} for definitions of special cases and other terms used in this section.} queries in $O(\log n \log \log n + k)$ time and $O(n\log n)$ space.  Our data structure can be also modified to
answer 5-sided queries (i.e., orthogonal range reporting queries, such that the query rectangle is bounded on five sides)  within the same time and space bounds, see Theorem~\ref{theor:5sid}. This result can be also extended to the general case of four-dimensional range reporting without increasing the query time (but at the cost of increasing the space usage by $O(\log^3 n)$ factor).  Finally the result can be also generalized to higher dimensions: There is a data structure that uses $O(n\log^{d}n)$ space and supports $d$-dimensional orthogonal range reporting queries in $O(\log^{d-3}n\log\log n + k)$ time  for any constant $d\ge 4$ (see Theorem~\ref{theor:multidim}). Our data structure can also support  range emptiness queries (i.e., determine whether a query range is empty)  in time $O(\log^{d-3}n \log\log n)$.

This paper is structured as follows. We describe our main result in Sections~\ref{sec:topbot} and~\ref{sec:5sid}.   Our data structure supports the special case of range reporting queries, the 4d 5-sided queries. In Section~\ref{sec:general} we show how the data structure can be modified to support more general types of orthogonal range reporting queries.
The main idea of our approach is to replace a single range tree with a hierarchy of range trees with decreasing fan-out.  We construct shallow cuttings for groups of sibling nodes in each tree (Section~\ref{sec:5sid}). This  hierarchy of trees with shallow cuttings assigned to nodes is somewhat similar to the method used in~\cite{Nekrich21soda}.  However a direct application of this method would lead to polynomial space usage; we need to modify this approach in order to save space as described in Sections~\ref{sec:topbot} and~\ref{sec:5sid}.  Although our data structure requires an almost-linear number of shallow cuttings, it can be constructed in $n\log^{O(1)}n$ time as explained in Section~\ref{sec:construct}. Our construction algorithm is based on a non-trivial modification of previous work; we believe that this result can be  of independent interest.

\begin{table}[bt]
  \centering
  \begin{tabular}[h]{|c|c|c|c|} \hline
    Reference & Query Type & Query Time & Space Usage\\ \hline
\cite{Afshani08}              &  Dominance &  $\log^2 n + k$ & $n\log n$ \\
\cite{AfshaniAL09,AfshaniAL10}              &  Dominance & $(\log n/\log\log n)^2+ k$ & $n(\log n/\log \log n)$ \\
\cite{AfshaniAL12}         &  Dominance & $\log^{3/2}n + k$  & $n(\log n/\log\log n)$ \\ \hline
    New       &  Dominance & $\log n \log\log n + k$ & $n\log n$ \\ \hline
   \cite{Bentley79}+\cite{ChazelleG86}               &  General &  $\log^3 n+k $ &  $n \log^2 n$ \\
\cite{Afshani08}              &  General   & $\log^2 n + k$         & $n\log ^4n$ \\ 
\cite{AfshaniAL09,AfshaniAL10}        &  General   & $(\log n/\log\log n)^2 + k$         & $n(\log n/\log\log n)^4$ \\
\cite{AfshaniAL12}        &  General   & $\log^{3/2}n + k$         & $n(\log n/\log\log n)^4$ \\  \hline
     New      &  General   & $\log n \log\log n + k$         & $n\log^4 n$ \\   \hline      
  \end{tabular}
  \caption{4d Orthogonal Range Reporting in the PM Model: Previous and New Results}
  \label{tab:results}
\end{table}

It is also interesting to compare data structures in the PM model to state of the art in the RAM model. It is known that it is possible to achieve $O(\log \log n + k)$ query time for two- and three-dimensional range reporting in the RAM model. On the other hand, any $n\log^{O(1)}n$-space data structure supports four-dimensional range reporting queries in  $\Omega(\log n/\log\log n)$ time~\cite{Patrascu11}. Thus in the RAM model there is an almost  logarithmic gap between the complexity of  orthogonal range reporting in three and four dimensions. The result of this paper indicates that there is either no  gap between 3d and 4d  or this gap is significantly smaller in the case of the PM model.


\section{Preliminaries}
\label{sec:prelim}

Let $P$ be a set of $n$ points in four-dimensional space.
A \emph{dominance} range query is a special case of the orthogonal range query when the query range is a product of half-open intervals.  A four-dimensional dominance range query is a query of the form 
$q=\prod_{i=1}^{4} (-\infty, a_i]$. A \emph{5-sided} query is a query of the form $[a_1,b_1]\times \prod_{i=2}^4(-\infty,b_i]$. In the case of an \emph{emptiness} query, the goal is to report whether 
there is any point of $P$  inside the query region $q$.

\paragraph{Shallow cuttings.} Consider two 3d points $p$ and $q$. 
A point $p$ dominates another point $q$ if and only if 
$p$ has a higher coordinate value than $q$ in each dimension.
Let $P$ be a set of $n$ points in 3d.
A $t$-level shallow cutting of $P$~\cite{Afshani08} 
is a collection ${\cal C}$ of {\em cells or boxes} of the form $(-\infty, a] \times (-\infty, b] \times (-\infty,c]$
 with the following three properties:
\begin{enumerate}
\item The number of cells is only $O(n/t)$, i.e., $|{\cal C}|=O(n/t)$.
\item Each box in  ${\cal C}$ dominates at most $c_2t$ points in $P$.
\item Any point in 3d which dominates at most $c_1t$ points in $P$ will lie 
inside at least one cell in ${\cal C}$.
\end{enumerate}

The {\em conflict list} of a cell $C \in {\cal C}$ is defined as the points of 
$P$ that are inside $C$. The {\em apex point} of a cell $(-\infty, a] \times (-\infty, b] \times (-\infty,c]$ is defined as $(a,b,c)$. 
Given a query point $q$ in 3d, the FIND-ANY query either reports a cell in ${\cal C}$
which contains $q$, or reports that there is no cell containing $q$. We will use 
the following fact about FIND-ANY queries.

\begin{lemma}\label{lem:find-any}
(FIND-ANY query~\cite{Afshani08}.) Let ${\cal C}$ be collection of cells in a $t$-level shallow cutting of $P$. There is a data structure of size $O(|{\cal C}|)$ which can answer the FIND-ANY query in 
$O(\log |{\cal C}|)$ time, where $|{\cal C}|$ denotes the number of cells in ${\cal C}$.
\end{lemma}

\section{Top and bottom structures}
\label{sec:topbot}

We use the notation  $[m]=\{1,2,\ldots,m\}$.
In a {\em restricted} 4d 5-sided  reporting query, 
the additional constraint is that 
first coordinates of points are from $[n^{1/3}]$ and for each $j\in [n^{1/3}]$
there are $n^{2/3}$ points of $P$ whose
first coordinate  value is equal to $j$.
 We will prove the following result about restricted 5-sided queries.

\begin{theorem}\label{thm:5-sided-assume}
Assume that there is a  restricted 4d 5-sided  reporting structure that  occupies $O(n\log n)$ space and
answers a query in $O(\log n\cdot\log\log n + k)$ time.
Then there is 4d 5-sided range reporting structure 
that occupies $O(n\log n)$ space 
and answers queries in 
$O(\log n\cdot\log\log n+ k)$ time.
\end{theorem}

\paragraph{Structure.} Let ${\cal T}$ be a range tree built on the 
first coordinates of the 
points in $P$. The fanout or degree of 
${\cal T}$ is two. 
Let $\ell=\frac{1}{3}\log n$ be 
the level of the tree (root is at level zero) 
which has $n^{1/3}$ nodes. Then the {\em top structure}, ${\cal T}_{top}$,  
is a subset of ${\cal T}$ restricted to 
levels $1, 2,\ldots, \ell$. 
Let $v_1, v_2,\ldots, v_t$ be the nodes in the left-to-right ordering 
at level $\ell$, 
and let $P(v_1), P(v_2),\ldots, P(v_t)$ be the 
points of $P$ lying in their subtrees, respectively. 
Then, for each point $p(p_1,p_2,p_3,p_4)\in P$, 
if $p$ lies in the subtree of 
$v_i$, then $p$ is transformed to a new point $p'=(i, p_2,p_3,p_4)$.
Let $P_{top}=\bigcup_{p\in P} \{p'\}$ be the new collection of points.
An instance of  restricted 4d 5-sided reporting structure (Theorem~\ref{thm:rest-5-sided}) 
on $P_{top}$ is constructed. 
Finally, recursive {\em bottom} structures are constructed on 
 $P(v_1), P(v_2),\ldots, P(v_t)$.
The recursion stops when the cardinality of a 
point set falls below a suitably large constant, 
and the queries on such point sets can be 
handled via naive scan.

\paragraph{Query algorithm: 5-sided queries.} The leaf nodes of ${\cal T}_{top}$ are the 
nodes at level $\ell=\frac{1}{3}\log n$.
For each leaf node $v\in {\cal T}_{top}$, 
its {\em range} $[a(v), b(v)]$ is defined as 
follows: $a(v)$ (resp., $b(v)$) is the 
first coordinate value of the 
leftmost (resp., rightmost) point in the
subtree of $v$. Let ${\cal R}_{top}$ be a 
collection of these ranges. 
Define $[a_1^{top}, b_1^{top}]$ 
to be the union of all the ranges in ${\cal R}_{top}$
which lie completely inside $[a_1,b_1]$.
Let $v_a$ and $v_b$ be the leaf nodes in ${\cal T}_{top}$ whose ranges contain
$a_1$ and $b_1$, respectively.
\begin{enumerate}
    \item If $v_a \neq v_b$, then perform a  restricted 4d 5-sided query
    on $P_{top}$ with the query range 
    $[a_1^{top}, b_1^{top}] \times 
    \prod_{i=2}^{4} (-\infty, b_i]$. 
    Next, perform 4d dominance queries
    on $P(v_a)$ and $P(v_b)$  with queries
    $[a_1, \infty) \times \prod_{i=2}^{4} (-\infty, b_i]$ and $(-\infty, b_1] \times 
    \prod_{i=2}^{4} (-\infty, b_i]$, respectively. Dominance queries can be answered as explained below. 
    
    \item If $v_a=v_b$, then perform a 4d 5-sided 
    query on the bottom structure corresponding 
    to $v_a$ with the query $q=[a_1,b_1] \times \prod_{i=2}^{4}(-\infty, b_i]$.
\end{enumerate}

\paragraph{Query algorithm:  4d dominance queries.}
Now we explain how a dominance query can be answered. First, the query algorithm performs a restricted  query 
on $P_{top}$ with the query range $(-\infty, b_1^{top}) 
\times \prod_{i=2}^{4} (-\infty, b_i]$, where $b_1^{top}$ is defined as above. Since a restricted 4d dominance query is a special case 
of a restricted 4d 5-sided query, we can use the data structure on $P_{top}$ to answer this query.  Let $v_b$ be the leaf node 
in ${\cal T}_{top}$ containing $b_1$. Then, we perform a 
4d dominance query on the bottom structure corresponding to 
$v_b$.

\paragraph{Analysis.} Let $S(n)$ be the space 
occupied by the data structure when built on 
$n$ points. At ${\cal T}_{top}$, we assume that  
the  restricted 4d 5-sided structure 
occupies $O(n\log n)$ space.
Then (ignoring the output size term), 
\[S(n) \leq n^{1/3}\cdot S(n^{2/3}) + O(n\log n). \]
Let $s(n)=S(n)/n$. Then,

\[ s(n) \leq s(n^{2/3}) + O(\log n) \implies 
s(n) = O(\log n).\]
This solves to $S(n)=O(n\log n)$.
Now we analyze the query time of dominance queries.
Let $Q_4(n)$ be the time taken to answer a 4d dominance 
query. Then,
\[ Q_4(n) = Q_4(n^{2/3}) + O(\log n\cdot\log\log n)=
O(\log n\cdot\log\log n).\]
Now we analyze the time $Q_5(n)$ of answering a 5-sided query. There are two cases:
\begin{enumerate}
\item If $v_a \neq v_b$,
then let $Q'_5(n)$ be the time taken to answer 
the restricted query at ${\cal T}_{top}$ (ignoring the output size term).
By assumption of Theorem~\ref{thm:5-sided-assume},  $Q'_5(n)=O(\log n\cdot\log\log n + k)$ 
time. We already know that $Q_4(n)=O(\log n\log\log n)$.  Hence, $Q_5(n)=O(Q_4(n)+Q'_5(n))=O(\log n\cdot\log\log n)$. 

\item If $v_a=v_b$, then let $Q_5(n)$ be the 
time taken to answer the query (again ignoring the 
output size term). Then, 
\[Q_5(n) = Q_5(n^{2/3}) + O(\log n). \]
\end{enumerate}

Note that
there is no recursion involved in the first case.
If the second case happens $i$ times before the first case happens, then
\begin{align*}
Q_5(n)&\leq O(\log n + \log n^{2/3} +\ldots + 
\log n^{(2/3)^i}) + Q'_5(n^{(2/3)^i}) +Q_4(n^{(2/3)^i})\\
&=O(\log n) + Q'_5(n^{(2/3)^i})+Q_4(n^{(2/3)^i})
\leq O(\log n) + Q'_5(n)+Q_4(n)=O(\log n\cdot\log\log n).
\end{align*}
This finishes the proof of Theorem~\ref{thm:5-sided-assume}.

\section{Restricted 4d 5-sided structure}
\label{sec:5sid}

In the previous section, Theorem~\ref{thm:5-sided-assume} 
 assumes the existence of 
a restricted 4d 5-sided structure, with $O(n\log n)$ space and 
$O(\log n\cdot\log\log n+k)$ query time. 
In this section we  build this data structure. 
 Recall that in a 
 restricted 4d 5-sided query, the input is a set $P$ of 
$n$ points with the restriction that, for each $j\in [n^{1/3}]$, 
there will be $n^{2/3}$ points of $P$ whose
first coordinate  value is equal to $j$. The query 
is a 5-sided box $[a_1,b_1] \times \prod_{i=2}^{4}(-\infty, b_i]$.

\begin{theorem}\label{thm:rest-5-sided}
There is a restricted 4d 5-sided reporting structure which occupies $O(n\log n)$ 
space and answers a query in $O(\log n\cdot\log\log n + k)$ time.
\end{theorem}

Combining the above  theorem (Theorem~\ref{thm:rest-5-sided}) with Theorem~\ref{thm:5-sided-assume} 
proves our main result about 4d 5-sided queries. 
\begin{theorem}
\label{theor:5sid}
There exists a  data structure that answers 5-sided 4d range reporting queries on a set $P$ in $O(\log n \cdot \log \log n + k)$ time and uses $O(n\log n)$ space, where $n$ is the number of points in $P$ and $k$ is the number of points in the query range.
\end{theorem}

We will start by describing the solution for the {\em emptiness} query, 
which reports whether $P \cap q$ is empty or not. Later, we will extend 
our solution to the reporting version of the problem, i.e., reporting 
all the points in $P \cap q$. 

\subsection{The emptiness query}

\paragraph{High-level overview.} A straightforward way to answer 
4d dominance emptiness query is to build a range tree on the 
first coordinate of the points in $P$, and reduce the query algorithm to 
solving $O(\log n)$ instances of FIND-ANY query (Lemma~\ref{lem:find-any}) 
which takes $O(\log^2n)$ time in total. Roughly speaking, Afshani et al.~\cite{AfshaniAL12} improved the 
query time to $O(\log^{1.5}n)$ by reducing the problem to an instance 
of 3d rectangle stabbing (where the input is axis-parallel boxes in 3d and the query is a 
point) which helps in answering  several instances of 
FIND-ANY query in one-shot. 

Our solution is based on a different approach. In the $i$-th iteration of the 
query algorithm, we  reduce the problem of deciding whether 
$P\cap q$ is empty or not to $O(3^i)$ instances of FIND-ANY queries
on shallow cuttings consisting of $n^{O(1/3^{i})}$ cells. As a result, by Lemma~\ref{lem:find-any},
the time spent in the $i$-th iteration will be $O(3^i\cdot \log n^{1/3^{i}})=O(\log n)$.
In total, only $O(\log_{3}\log n)$ iterations will be performed and hence, the query 
time is reduced to $O(\log n\cdot \log_{3}\log n)$. Our solution for restricted 4d 5-sided query
is a generalization of this approach. 

Consider a  parameter $\beta$. 
Some of the equations in this section will be presented 
in terms of $\beta$.  Setting $\beta=3$ 
will correspond to the restricted 
4d 5-sided solution.

\paragraph{Idea 1: Constructing many range trees.}
We start by constructing $\log_{\beta}\log n$ trees $T_2,\ldots, T_{\log_{\beta}\log n}$.
The fanout, $f_i$, of tree $T_i$ is $n^{1/\beta^i}$. 
Each $T_i$ has $n^{1/3}$ leaf nodes and the $j$-th leaf node
stores all the points in $P$ whose first coordinate value is equal to $j$.
For $2\leq i < \log_{\beta}\log n$,  
tree $T_{i+1}$ can be constructed from $T_{i}$ as follows:
pick an internal node $v\in T_i$ and let $Ch(v)$ be its child nodes.
We will replace $v$ and $Ch(v)$ (a two-level tree) with a subtree 
of fanout $f_{i+1}$, the root of the new subtree is the  node $v$ and $Ch(v)$ 
are the leaf nodes of the new subtree. After each internal node in $T_i$ performs this operation,  
the tree $T_{i+1}$ is constructed.

\begin{figure}[bt]
 \begin{center}
      \includegraphics[scale=1]{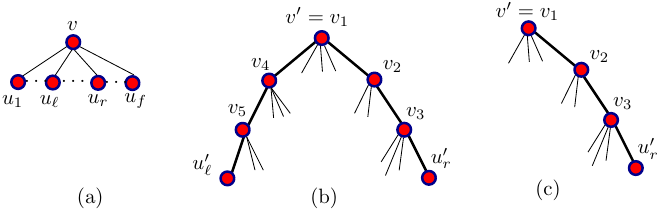}
 \end{center}
 \caption{(a) A node $v \in T_i$, (b) For $\beta=3$, five nodes in $T_{i+1}$ are enough
 to cover any bounded range of $v$,  and (c) For $\beta=3$, three nodes in $T_{i+1}$ 
 are enough to cover any prefix range of $v$.}
 \label{fig:expansion}
\end{figure}

\begin{lemma}\label{lem:cover-points}
\textbf{(Covering ranges of points)} For $i< \log_\beta\log n$, let $v$ be any internal node in tree $T_i$ 
and let 
$u_1, u_2,\ldots, u_f$ be its $f$ children from left-to-right. 
For $\ell \leq r$, define a ``bounded range'' $P(v, \ell, r) \subseteq P$ to be the points 
in the subtrees of $u_{\ell}, u_{\ell+1},\ldots, u_r$.
A ``prefix range'' (resp., ``suffix range'') is a special case of 
bounded range with $\ell=1$ (resp., $r=f$).
 Then, 
\begin{itemize}
    \item (Covering bounded ranges) For any $P(v, \ell, r)$, there exists 
    at most $(2\beta-1)$ nodes $v_1,v_2,\ldots,v_{2\beta-1}$ in tree $T_{i+1}$ and
    integers $\ell_1,\ldots, \ell_{2\beta-1}, r_1, \ldots, r_{2\beta-1}$ such that 
    \[P(v, \ell, r)= \bigcup_{j\in [2\beta-1]} P(v_j,\ell_j,r_j). \]
    Among the $(2\beta{-}1)$ bounded ranges, 
    at least $(2\beta{-}2)$ are  either 
    prefix or suffix ranges.
	\item (Covering prefix and suffix ranges)
    For any $P(v,1, r)$, there exists 
    at most $\beta$ nodes $v_1,v_2,\ldots,v_{\beta}$ in tree $T_{i+1}$ and
    integers $r_1, \ldots, r_{\beta}$ such that 
    \[P(v, 1, r)= \bigcup_{j\in [\beta]} P(v_j,1,r_j). \]
    This holds analogously for any suffix range.
\end{itemize}

\end{lemma}
See Fig.~\ref{fig:expansion} for an example of Lemma~\ref{lem:cover-points}; the proof is deferred  to Section~\ref{sec:app-proof}.

\paragraph{Idea 2: Nested shallow cuttings for all bounded 
ranges.} 
Next, multiple shallow cuttings will be constructed at each node in trees $T_2,\ldots, T_{\log_{\beta}\log n}$.
Recall that $f_i=n^{1/\beta^i}$ is the fanout of tree $T_i$, and define $f_1=n^{1/\beta}$.
Consider an arbitrary node $v \in T_i$. 
For all $1\leq \ell \leq r\leq f_i$, a $f_{i-1}^2$-level shallow cutting, ${\cal C}(v,\ell,r)$, is constructed  based on the  last three coordinates of the pointset $P(v, \ell, r)$. We will say that ${\cal C}(v,\ell,r)$  and its cells are \emph{associated} with a tree $T_i$

Let ${\cal C}$ be the collection of cells from all the shallow cuttings constructed.
Pick any cell $C \in {\cal C}$. If $C$ is associated with tree $T_i$, then 
a $(f_i^2/c)$-level shallow cutting is constructed 
based on the conflict list of $C$, where $c$ is a sufficiently large constant. 
The resulting shallow cuttings and cells are called 
{\em nested shallow cuttings} and {\em nested cells}, respectively.
 A FIND-ANY structure of 
Lemma~\ref{lem:find-any}
is constructed based on these nested cells.
Such nested shallow cuttings 
are constructed for each cell in ${\cal C}$. Let ${\cal C}_N$ be the collection of 
all the nested cells constructed. 
The data structure will store all the cells in ${\cal C}$ and ${\cal C}_N$, and 
the FIND-ANY structures corresponding to each cell in ${\cal C}$.
Additionally, for all the nested cells in ${\cal C}_N$ which were constructed at 
the last tree $T_{\log_{\beta}\log n}$, their conflict lists are explicitly stored.

\begin{lemma}
  \label{lemma:size}
The size of the data  structure is $O(n\log n)$.
\end{lemma}
\begin{proof}
Consider a node $v \in T_i$. Firstly, 
$|P(v)| \geq n^{2/3}$, since each leaf node in $T_i$ is associated with 
at least $n^{2/3}$ points. Secondly, for $\beta=3$, we have
$f_{i-1}^2 \leq f_1^2 \leq n^{2/\beta}=n^{2/3}$.
Therefore, $|P(v)| =\Omega(f_{i-1}^2)$. As a result, 
a $f_{i-1}^2$-level shallow cutting on $P(v, \ell, r)$ will have $O(|P(v)|/f_{i-1}^2)$ cells. 
Since $O(f_i^2)$ such shallow cuttings are constructed at $v$, the total number of cells constructed 
will be $O(|P(v)|\cdot f_i^2/f^2_{i-1})$.
The number of cells constructed 
at all the nodes of a given level will be $O(\sum_v |P(v)|\cdot f_i^2/f_{i-1}^2)=O(n\cdot f_i^2/f_{i-1}^2)$.
 Since the conflict list size of each cell is 
$O(f_{i-1}^2)$, a $(f_i^2/c)$-level nested cutting on this conflict list creates $O(f_{i-1}^2/f_i^2)$ nested cells.
As such, 
the number of nested cells at a given level is $O(n\cdot f_i^2/f^2_{i-1}) \times O(f_{i-1}^2/f_i^2)
=O(n)$.  For $\beta=3$, the height of tree $T_i$ is $(\log n^{1/3})/(\log n^{1/3^i})=3^{i-1}$.
Therefore, the number of nested cells  in a tree $T_i$ will be $O(n.3^i)$. 
Summing over all the trees, the number of nested cells 
is $\sum_{i=0}^{\log_3\log n} O(n\cdot 3^i)=O(n\log n)$. 

In tree $T_{\log_{3}\log n}$ the conflict list is explicitly stored with each nested cell. 
The fanout of this tree is two and hence, the size of the conflict list of each 
nested cell is $O(1)$. The number of nested cells in this tree is 
$O(n\log n)$. As a result, storing the conflict list increases 
the space by only a constant factor, and hence, the overall space remains $O(n\log n)$.
\end{proof}

\begin{figure}[tb]
 \begin{center}
      \includegraphics[scale=0.7]{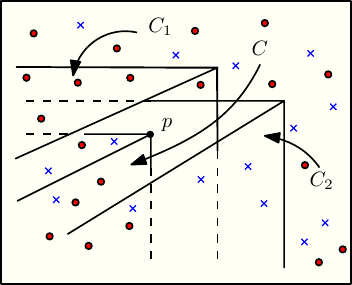}
 \end{center}
 \caption{$C$ is a nested cell constructed 
 at $T_i$ and is associated with a prefix range.
 $\Gamma(C)=\{C_1, C_2, C_3\}$ for $\beta=3$. 
 The figure shows only $C_1$ and $C_2$ and points from $v_1$ and $v_2$ (denoted by red circles and blue crosses respectively).}
 \label{fig:connecting-cells}
 \end{figure}
 
\paragraph{Idea 3: Connecting adjacent trees via cells and nested cells.} In the following lemma, we establish 
the connection between nested cells of $T_i$ and cells of adjacent tree $T_{i+1}$. Henceforth for any cell (or nested cell) $C$, we will denote by $L(C)$ its conflict list.

\begin{lemma}\label{lem:beta-cells}
Consider any nested cell $C\in {\cal C}_N$ which is associated with $T_i$, for $2\leq i <\log_{\beta}\log n$.
Then there exists a set $\Gamma(C)  \subseteq {\cal C}$ such that
\begin{enumerate}
    \item Each cell in $\Gamma(C)$ is associated with tree $T_{i+1}$, 
    \item For all $C' \in \Gamma(C)$, we have $C \subseteq C'$, and 
    \item if a region $q'=\prod\limits_{i=2}^{4}(-\infty, b_i]$ lies inside 
    $C$, then $L(C) \cap q'= \bigcup_{C' \in \Gamma(C)}(L(C') \cap q')$.
\end{enumerate}
If $C$ is associated with a bounded range  $P(v, \ell, r)$, then 
$|\Gamma(C)| \leq 2\beta-1$, and at most one cell in $\Gamma(C)$ is associated 
with a bounded range. Otherwise, if $C$ is associated with a 
prefix or a suffix range, then 
$|\Gamma(C)| \leq \beta$, and all the cells in $\Gamma(C)$ are associated 
with a prefix or a suffix range.
\end{lemma}

\begin{proof}
Define $\widehat{C}$ to be the cell on whose 
conflict list, $L(\widehat{C})$, the $(f_i^2/c)$-level nested 
cutting led to the construction of nested cell $C$.
Let $v$ be the node in $T_i$ such that $\widehat{C}$ is a cell of the shallow cutting for some pointset   $P(v,\ell,r)$. Consider the following two cases:
\begin{itemize}
    \item If $P(v, \ell, r)$ is a bounded range, then via the first case in 
    Lemma~\ref{lem:cover-points}, we define $N(C)=\{v_1, v_2,\ldots, v_{2\beta-1}\}$, and 
    \item If $P(v, \ell, r)$ is a prefix or a suffix range, then via the second case in 
    Lemma~\ref{lem:cover-points}, we define $N(C)=\{v_1, v_2,\ldots, v_{\beta}\}$.
\end{itemize}

In either case, 
\begin{equation}\label{eqn:beta-points}
P(v, \ell, r)= \bigcup_{v_j \in N(C)} P(v_j,\ell_j,r_j).
\end{equation}

Let $p$ be the apex point of $C$. Since $L(\widehat{C})= \widehat{C} \cap P(v,\ell,r)$, 
the number of points of $L(\widehat{C})$ or $P(v,\ell,r)$ which are dominated 
by $p$ will be less than or equal to $f_i^2$, since we chose a sufficiently large constant $c$ 
for the $(f_i^2/c)$-level nested shallow cutting on $L(\widehat{C})$. 
For any $v_j \in N(C)$, since $P(v_j,\ell_j,r_j) \subseteq P(v,\ell, r)$, 
the number of points of $P(v_j,\ell_j,r_j)$ 
dominated by $p$ will be less than or equal to $f_i^2$.
For all $v_j \in N(C)$, the  data structure constructs a $f_i^2$-level 
shallow cutting on $P(v_j, \ell_j, r_j)$. 
Therefore, for each $v_j \in N(C)$, there exists a cell, say $C_j$,
in the $f_i^2$-level shallow cutting constructed on $P(v_j, \ell_j, r_j)$ which contains $p$. 
Then define $\Gamma(C)=\bigcup_{v_j \in N(C)} \{C_j\}$.
This establishes items (1) and (2) in Lemma~\ref{lem:beta-cells}. See Figure~\ref{fig:connecting-cells}.

Now we establish (3). Since $q'\subseteq C \subseteq C_j$, for all 
$C_j \in \Gamma(C)$, we infer  that 
\begin{equation}\label{eqn:subset}
q'= q'\cap C =q' \cap C_j, \text{ for any $C_j \in \Gamma(C)$}.
\end{equation}
From Equation~\ref{eqn:beta-points}, we claim that 
\begin{align*}
& P(v,\ell,r) \cap q' =  \bigcup_{v_j \in N(C)} (P(v_j,\ell_j,r_j) \cap q') \\
\implies & P(v,\ell,r) \cap (q' \cap C) = \bigcup_{v_j \in N(C)} (P(v_j,\ell_j,r_j) \cap (q'\cap C_j ))
\quad \text{(From equation~\ref{eqn:subset})}\\
\implies & L(C) \cap q' = \bigcup_{C_j\in \Gamma(C)} (L(C_j) \cap q').
\end{align*}
\end{proof}

Using the above lemma, the data structure stores pointers from each nested cell  $C \in {\cal C}_N$ to 
its corresponding cells in $\Gamma(C)$.

\paragraph{Query algorithm.} 
Recall that $q'=\prod_{j=2}^{4}(-\infty, b_j]$.
Starting from $i{=}1$, in the $i$-th iteration of the 
query algorithm, a collection of cells ${\cal C}_{i+1}$ is maintained. 
At the beginning of the $i$-th iteration, the query algorithm will satisfy the 
following three invariants:
\begin{enumerate}
    \item The cells in 
${\cal C}_{i+1}$ are associated to  the tree $T_{i+1}$.

\item $q'$ is guaranteed to lie inside 
all the cells in ${\cal C}_{i+1}$.

\item The problem of determining 
whether $|P\cap q| \geq 1$ is reduced 
to the problem to 
determining whether $|L(\widehat{C}) \cap q'| \geq 1$, 
for any $\widehat{C}\in {\cal C}_{i+1}$.
\end{enumerate}

During the $i$-th iteration, for each cell 
$\widehat{C}\in {\cal C}_{i+1}$, 
we do the following: Perform FIND-ANY query on 
the nested shallow cutting constructed 
based on $L(\widehat{C})$. If no nested cell 
is reported, then we conclude that $|P\cap q|\geq 1$ and stop the algorithm. 
Otherwise, let $C$ be the cell reported. 
Based on Lemma~\ref{lem:beta-cells}, add 
the cells in $\Gamma(C)$ to ${\cal C}_{i+2}$.

In the $i$-th iteration,  
if the algorithm does not stop, then
the three invariants continue to hold 
at the beginning of the $(i{+}1)$-th iteration. 
Bullet~(1) of Lemma~\ref{lem:beta-cells} ensures
that 
all the cells in ${\cal C}_{i+2}$ are associated to 
$T_{i+2}$. For each cell $\widehat{C} \in 
{\cal C}_{i+1}$,  its corresponding nested cell 
$C$ contains $q'$. Combining this fact with
bullet~(2) of Lemma~\ref{lem:beta-cells}
ensures that $q'$ lies inside all cells in ${\cal C}_{i+2}$.
Since $L(\widehat{C}) \cap q= 
L(C) \cap q$, combining this fact with
bullet~(3) of Lemma~\ref{lem:beta-cells} 
ensures that 
computing $L(\widehat{C})\cap q$ is equivalent 
to 
computing $\bigcup_{C' \in \Gamma(C)}(L(C') \cap q)$.

In the end, if the algorithm computes ${\cal C}_{\log_{\beta}\log n}$, then we stop the recursion and 
explicitly scan the conflict list of all the cells in ${\cal C}_{\log_{\beta}\log n}$
to check if any point lies inside $q$.

{\em Initialization steps.}
As an initialization step  we query $T_2$
with $[a_1,b_1]$. Let $\pi_{\ell}$ and $\pi_r$
be the path from the root node to the leaf nodes 
corresponding to $a_1$ and $b_1$, respectively.
For each node $v\in \pi_{\ell} \cup \pi_r$, 
let $w_1, w_2,\ldots, w_f$ be its children.
Let $w_{\ell}$ (resp., $w_r$) be the leftmost 
(resp., rightmost) child such that the first coordinate value 
of all the points in $P(w_{\ell})$
(resp., $P(w_r)$) lie inside $[a_1,b_1]$. 
Then perform FIND-ANY query on the shallow
cutting constructed on $P(v,\ell,r)$. If no 
cell is reported, then conclude that $|P\cap q|
\geq 1$, and stop the algorithm. Otherwise, 
add the cell reported to ${\cal C}_2$. 
Since the height of $T_2$ is a constant, 
 $|{\cal C}_2| \leq c_0$, where $c_0$ is a sufficiently 
 large constant. 

\begin{lemma}\label{lem:size-of-c}
For $i\geq 2$, $|{\cal C}_i| \leq c_0\cdot \beta^{i+1}-1$, 
where $c_0$ is a sufficiently large constant.
\end{lemma}
\begin{proof}
For  all $2\leq i \leq \log_{\beta}\log n$, 
at most one cell in ${\cal C}_i$ is associated with a bounded range. 
We prove it via induction. For ${\cal C}_2$, the cell in ${\cal C}_2$ corresponding to
the root node of $T_2$ is associated with a bounded range (see Figure~\ref{fig:expansion}(b)).
The remaining cells in ${\cal C}_2$ are associated with a prefix or a suffix range.
Assume it is true for ${\cal C}_i$. Let $\widehat{C}$ be the only cell in ${\cal C}_i$ associated 
with a bounded range, and let $C$ be the nested cell reported by the FIND-ANY query on $L(\widehat{C})$.
Via Lemma~\ref{lem:beta-cells}, observe that at most one cell in $\Gamma(C) \subseteq {\cal C}_{i+1}$ 
is 
associated with a bounded range. For any other cell $\widehat{C} \in {\cal C}_i$ and its corresponding 
nested cell $C$, all the cells in $\Gamma(C) \subseteq {\cal C}_{i+1}$ are associated with a 
prefix or a suffix range (via Lemma~\ref{lem:beta-cells}). Therefore, we conclude that at most one cell in ${\cal C}_{i+1}$ 
is associated with a bounded range. This leads to the following recurrence:
\begin{align*}
|{\cal C}_{i+1}| &\leq \beta \cdot (|{\cal C}_i|-1) + (2\beta -1) \quad \text{ via Lemma~\ref{lem:beta-cells}} \\
&=\beta\cdot|{\cal C}_i| +\beta-1\\
&\leq \beta(c_0\cdot\beta^i-1 ) +\beta-1 \quad \text{ via induction} \\
&\leq c_0\cdot \beta^{i+1}-1.
\end{align*}
\end{proof}

For any cell $C \in {\cal C}_i$, 
 the number of nested cells in the nested 
shallow cutting constructed on $L(C)$ is
$O(|L(C)|/f_i^2)=O(f_{i-1}^2/f_i^2)=O(n^{2/\beta^{i-1}}/n^{2/\beta^i}) 
=O(n^{\frac{2\beta-2}{\beta^i}})$. Hence, performing a
FIND-ANY operation on the nested shallow cutting
takes $O(\frac{2\beta-2}{\beta^i}\cdot\log n)$ time. As such, 
performing FIND-ANY operations w.r.t. all the 
cells in ${\cal C}_i$ takes 
$O(|{\cal C}_i|\cdot \frac{2\beta-2}{\beta^i}\cdot\log n)
=O(\log n)$ time, since Lemma~\ref{lem:size-of-c} states that 
$|{\cal C}_i|=O(\beta^i)$. Overall, 
it takes $O(\log n\cdot\log_{\beta}\log n)$ time 
to generate sets ${\cal C}_1, {\cal C}_2, \ldots, {\cal C}_{\log_{\beta}\log n}$. The time taken 
to explicitly scan the conflict lists of 
cells in ${\cal C}_{\log_{\beta}\log n}$ is $O(\log n)$, 
since $|{\cal C}_{\log_{\beta}\log n}|=O(\log n)$, 
and each conflict list has a constant size.
The initialization step can be performed 
in $O(\log n)$ time. Overall, 
the query time is $O(\log n\cdot \log\log n)$.

\subsection{The reporting query}
To handle the reporting query, we will need a 
 few additional data structures. 
Firstly, for each node $v\in T_{\log_{\beta}\log n}$, 
based on the points in the subtree of $v$, 
construct a $(c\log^6n)$-level shallow cutting, 
where $c$ is a sufficiently large constant.
Let ${\cal C}_{\log^6n}$ be the collection 
of all the cells associated with each node in 
$T_{\log_{\beta}\log n}$.
For each cell in the cutting, based on its 
conflict list build a data structure which
answers 3d dominance queries (on the last three coordinates). 
Secondly, 
construct a {\em slow} data structure on $P$ to 
answer 4d dominance queries: The data structure of Afshani~\cite{Afshani08} for 3d dominance queries
along with a range tree on the fourth coordinate leads to a data structure with $O(n\log n)$ size 
and $O(\log^2n + k)$ query time. 
The third data structure requires storing additional 
pointers across trees which is discussed next.

\begin{observation}
Let $\Delta=\log_{\beta}\log n -\log_{\beta}\log\log n$. Then, 
$\beta^{\Delta}=\frac{\log n}{\log\log n}$.
\end{observation}

\begin{lemma}\label{lem:gen-cell-cover}
For a cell $C$ associated with tree $T_{\Delta}$,
there exists a set $\Lambda(C) \in {\cal C}_{\log^6n}$
of at most $O(\log\log n)$ cells such 
that if $q'$ lies inside $C$, then 
$L(C) \cap q'= \bigcup_{C'\in \Lambda(C)}(L(C')\cap q')$.
\end{lemma}

\begin{proof}
Consider the case where $C$ is associated with a bounded range at a node $v\in T_{\Delta}$.
Let $P(v,\ell,r)$ be the pointset on which the shallow
cutting led to the construction of $C$.
After applying first case of Lemma~\ref{lem:cover-points},
let $N_1(C)$ 
be the collection of {\em covering nodes} in $T_{\Delta+1}$   such that
$P(v,\ell, r)= \bigcup_{v_j\in N_1(C)}P(v_j, \ell_j, r_j)$.
Then, $|N_1(C)|\leq 2\beta-1$. Now, let $N_2(C)$ be 
the collection of covering nodes in $T_{\Delta+2}$ after applying 
Lemma~\ref{lem:cover-points} on  pointset 
$P(v_j, \ell_j,r_j)$, for all $v_j\in N_1(C)$. 
At most one node in $N_1(C)$ corresponds to a pointset 
with a bounded range. Therefore, 
\[|N_2(C)| \leq \beta(|N_1(C)|-1) + (2\beta-1) \leq \beta (2\beta-2) + (2\beta-1) = 2\beta^2-1.\]
Performing this recursively, we obtain 
\[|N_i(C)| \leq 2\beta^i-1, \text{ for any $i\geq 2$.}\]
In particular, we obtain $|N_{\log_{\beta}\log\log n}(C)|=O(\log\log n)$, and 
\begin{equation}\label{eqn:gen-pointsets}
P(v,\ell, r)= \bigcup_{v_j\in N_{\log_{\beta}\log\log n}(C)}P(v_j, \ell_j, r_j).
\end{equation}

Now we adapt the proof of Lemma~\ref{lem:beta-cells}
to finish the proof. 
Let $p$ be the apex point of $C$. Since $L(C)= C \cap P(v,\ell,r)$, 
the number of points of $L(C)$ or $P(v,\ell,r)$ which are dominated 
by $p$ will be  $\Theta(\log^{2\beta}n)$, 
since we construct an $f_{\Delta-1}^2$-level shallow cutting on $P(v, \ell, r)$, 
and $f_{\Delta-1}^2= (n^{1/\beta^{\Delta-1}})^2=\left(n^{\frac{\beta \log\log n}{\log n}} \right)^2
=\Theta(\log^{2\beta}n)$.
For any $v_j \in N_{\log_{\beta}\log\log n}(C)$, 
since  $P(v_j,\ell_j,r_j) 
\subseteq P(v,\ell,r)$, 
the number of points of $P(v_j, \ell_j,r_j)$ 
dominated by $p$ will be $O(\log^{2\beta}n)$.
The data structure constructs a $(c\log^6n)$-level shallow cutting on $P(v_j,\ell_j,r_j)$. 
Since $c$ is chosen to be a sufficiently large constant,
there exists a cell, say $C_j$, in the $(c\log^6n)$-level shallow cutting which contains $p$.
Then we set $\Lambda(C)=\bigcup_{v_j \in N_{\log_{\beta}\log\log n}(C)} \{C_j\}$.

For any $C_j \in \Lambda(C)$, since 
 $q'\subseteq C \subseteq C_j$,  we claim that 
\begin{equation}\label{eqn:gen-subset}
q'= q'\cap C =q' \cap C_j.
\end{equation}
From Equation~\ref{eqn:gen-pointsets}, we claim that 
\begin{align*}
&P(v,\ell,r) \cap q'= \bigcup_{C_j \in \Lambda(C)} (P(v_j, \ell_j,r_j)\cap q')\\
\implies & P(v, \ell, r) \cap (q'\cap C)= 
\bigcup_{C_j \in \Lambda(C)} (P(v_j, \ell_j,r_j) \cap (q' \cap C_j))
\quad \text{(From equation~\ref{eqn:gen-subset})}\\
\implies &L(C) \cap q'=\bigcup_{C_j \in \Lambda(C)}  (L(C_j) \cap q').
\end{align*}
\end{proof}
Using Lemma~\ref{lem:gen-cell-cover}, 
for each cell $C$ associated with tree 
$T_{\Delta}$ we will maintain pointers to 
each cell in $\Lambda(C)$.

\begin{lemma}
The space occupied by  additional 
data structures is $O(n\log n)$.
\end{lemma}
\begin{proof}
The space 
occupied by a single dominance structure 
will be $O(\log^6n)$ and the number of 
cells in a $(c\log^6n)$-level shallow cutting 
is $O(|P(v)|/\log^6n)$, where $|P(v)|$ is the number 
of points in the subtree of $v$. Therefore, 
the space occupied at any given level of 
$T_{\log_{\beta}\log n}$ is $O(\sum_{v} \frac{|P(v)|}{\log^6n}\times \log^6n)=O(n)$.
Since the height of $T_{\log_{\beta}\log n}$ is 
$O(\log n)$, the overall space occupied by all the 
3d dominance structures will be $O(n\log n)$.

The number of 
cells associated with a node $v$ in $T_{\Delta}$ 
(with fanout $f_{\Delta}=n^{1/\beta^{\Delta}}$)
is $O(|P(v)|/f_{\Delta-1}^2)\times O(f_{\Delta})=O(|P(v)|/f_{\Delta})$,
and as such, the number of cells at a 
given level in $T_{\Delta}$ is 
$O(n/f_{\Delta})$. Therefore, via Lemma~\ref{lem:gen-cell-cover}, 
the total number of pointers 
maintained will be 
\[ O\left(\frac{n}{f_{\Delta}}\cdot\log\log n\right)=
O\left(\frac{n}{n^{1/\beta^{\Delta}}}\cdot\log\log n\right)
= O\left(\frac{n\log\log n}{n^{\frac{\log \log n}{\log n}}}\right)=O\left(\frac{n\log\log n}{\log n}\right)=o(n\log n).\] 
\end{proof}

To answer a reporting query, we start by performing the same steps as the 
query algorithm for the emptiness query. Recall that if the emptiness query algorithm 
enters the $i$-th iteration, then either the algorithm stops in that iteration  or 
a collection ${\cal C}_{i+2}$ is generated for the $(i{+}1)$-th iteration.
To answer the reporting query,  two modifications are made to the emptiness query algorithm:
\begin{enumerate}
    \item If the algorithm stops in the $i$-th iteration, and $i<\Delta{-}1$, then 
    query the slow structure.
    \item If the algorithm reaches the $(\Delta{-}1)$-th iteration, then for each cell 
    $C \in {\cal C}_{\Delta}$, query with 
    $q'$ the 3d 
    dominance structures associated 
    with each cell in $\Lambda(C)$ (Lemma~\ref{lem:gen-cell-cover}).
\end{enumerate}

\begin{lemma}
The query time is $O(\log n\cdot\log\log n+ k)$.
\end{lemma}
\begin{proof}
 If case~(1) happens, then via properties of nested shallow cuttings, 
we conclude that $|P\cap q|=\Omega(f_{i+1}^2)=\Omega(n^{2/\beta^{i+1}})$, since an $f_{i+1}^2$-level 
nested shallow cutting 
is performed at cells associated with tree $T_{i+1}$. Since $i< (\Delta-1)$, it implies that 
$|P\cap q|=\Omega(n^{2/\beta^{\Delta-1}})
=\Omega(n^{2\beta\log\log n/\log n})=\Omega(\log^{2\beta}n)=\Omega(\log^2n)$, for $\beta=3$.
As a result, querying the slow data structure takes $O(\log^2n + k)=O(k)$ time.

If case~(2) happens, then the number of 
3d dominance structures queried will be 
$O(|{\cal C}_{\Delta}|\cdot \log \log n)$ (via Lemma~\ref{lem:gen-cell-cover}). 
The time taken to query a single 3d dominance 
structure is $O(\log\log^6n + k_i)
=O(\log\log n + k_i)$, where 
$k_i$ is the output size. Therefore, the 
overall query time will be 
$O(|{\cal C}_{\Delta}|\cdot \log \log n \cdot 
\log\log n + \sum_{i} k_i)=
O(\beta^{\Delta}\log^2\log n + k)=O(\log n\cdot\log\log n + k)$ (via Lemma~\ref{lem:size-of-c}).
\end{proof}

\section{General 4d orthogonal range reporting and higher dimensions}
\label{sec:general}
Our result for five-sided 4d queries can be extended to answer general 4d queries using standard techniques. For example, we can support queries $\prod_{i=1}^{2}[a_i,b_i]\times \prod_{i=3}^4(-\infty,b_i]$ by constructing the range tree $T_2$ on the second coordinates of all points. We keep the data structure for five-sided queries in every internal node of $T_2$. For a given query
$q=\prod_{i=1}^{2}[a_i,b_i]\times \prod_{i=3}^4(-\infty,b_i]$, we find the leaf that contains the predecessor of $b_2$, the leaf that contains the successor of $a_2$,  and their lowest common ancestor $w$.  Let $w_l$ and $w_r$ be the left and right children of $w$.  Let $q_l= [a_1,b_1]\times [a_2,+\infty)\times \prod_{i=3}^4(-\infty,b_i]$ and $q_r=[a_1,b_1]\times (-\infty,b_2]\times \prod_{i=3}^4(-\infty,b_i]$. Then $q\cap P=q\cap P(w)= (q_l\cap P(w_l))\cup (q_r\cap P(w_r))$. Since both $q_l$ and $q_r$ are five-sided queries, we can answer .  In the same way, we can support general queries $q=\prod_{i=1}^4[a_i,b_i]$ without increasing the query time and  using $O(n\log^4n)$ space.

We can also obtain a data structure that supports general  queries in $d>4$ dimensions using range trees. The space usage and query time grow by $O(\log n)$ factor with every dimension $d>4$.
\begin{theorem}
  \label{theor:multidim}
  There exists a data structure that answers $d$-dimensional orthogonal range reporting queries on a pointset $P$ in $O(\log^{d-3}n\log\log n + k)$ time and uses $O(n\log^dn)$ space, where $k$ is the number of points in the query range and $n$ is the number of points in $P$. 
\end{theorem}


\newcommand{\cQ}{{\cal Q}}
\newcommand{\cL}{{\cal L}}
\newcommand{\cD}{{\cal D}}
\newcommand{\cC}{{\cal C}}

\newcommand{\tcount}{\mathtt{count}}
\newcommand{\tselect}{\mathtt{select}}
\newcommand{\fin}{\mathtt{end}}

\section{Construction Algorithm}
\label{sec:construct}

As explained in previous sections, our data structure relies on a large number of three-dimensional shallow cuttings. There are algorithms that construct shallow cuttings in nearly linear time. That is, a shallow cutting of an $n$-point set can be constructed in time $O(n\log n)$~\cite{AfshaniT18,AfshaniCT14}. If we directly apply these algorithms, it would take   $\Omega(n^{3/2})$ time to construct our data structure.  

In this paper we use a different approach: We observe that conflict lists of cells in shallow cuttings are not needed in our case. The construction algorithm is modified so that the runtime is pseudo-linear in the \emph{number of cells} of the constructed 3d shallow cutting. We start by reviewing the algorithm of Afshani and Tsakalides~\cite{AfshaniT18}. Then we describe our implementation. Next we  show how the new  algorithm for 3d shallow cuttings can be applied to create our data structure for restricted 5-sided queries (Theorem~\ref{thm:rest-5-sided}) in $O(n\log^5n)$ time. Finally we explain how this result can be used to construct data structures for general orthogonal range reporting. In this section we will denote the second, third, and fourth coordinates of a point $p$ by $x(p)$, $y(p)$, and $z(p)$ respectively.

\paragraph{Previous Work: Overview of an $O(n\log n)$-Time Algorithm.}
In~\cite{AfshaniT18} the authors show how a 2d shallow cutting can be maintained under deletions. Then they combine their method with the sweep-plane approach and obtain a three-dimensional shallow cutting.  
A two-dimensional shallow cutting is a staircase, i.e., a polyline obtained by alternating horizontal and vertical segments, see Fig.~\ref{fig:staircase}.  An \emph{inner corner} of a 2d shallow cutting is the lower endpoint of a vertical segment (resp.\ the left endpoint of a horizontal segment) and an \emph{outer corner} is the upper endpoint of a vertical segment (resp.\ the right endpoint of a horizontal segment).    When some inner corner  dominates less than $k$ points, we start the process of \emph{patching} the staircase. Let $c_i$ denote the outer corner that precedes $d_i$. Let $c'_i$ denote the leftmost point such that $y(c_i)=y'(c_i)$ and $c'_i$ dominates $10k$ points. We extend the preceding outer corner to the point $c_i'$ and $c'_i$ becomes a new outer corner. Then we find the point $d'_{i+1}$, such that $x(d'_{i+1})=x(c'_i)$ and $d'_{i+1}$ dominates $9k$ points;  we add a vertical segment $[c_i',d'_{i+1}]$ to the staircase and $d'_{i+1}$ becomes a new inner corner. Next, we find the leftmost point $c'_{i+1}$ with  $x(c'_{i+1})=x(d'_{i+1}$ such that $c'_{i+1}$ dominates $10k$ points; we add the horizontal segment $[d'_{i+1},c'_{i+1}]$ to the staircase.  We continue adding horizontal and vertical segments with corresponding outer and inner corners to the staircase until the following termination condition is satisfied: Let $d_j$ be the rightmost inner corner dominated by $c'_i$. If $d_j$ dominates at least $7k$ points, we add the corner $c''_i$ instead of $c'_i$, so that $x(c''_i)=x(d_j)$ and $y(c''_i)=y(c'_i)$.  Finally we add a vertical segment $[c''_i,d_j]$ and finish the process of patching. All corners dominated by new outer corners are removed from the two-dimensional shallow cutting. 

In order to obtain a three-dimensional shallow cutting, a plane parallel to the $xy$-plane is moved in $-z$-direction starting at $z=+\infty$.
A two-dimensional shallow cutting is maintained for all points that are below the plane\footnote{To simplify the description, we sometimes do not distinguish between a point and its $xy$-projection.}.  Every time when the plane hits a point, we remove it from the set of points below the plane. When some inner corner of the staircase dominates less than $k$ points, we patch the staircase as described above.  Every added  outer corner becomes an apex point of some cell in the  three-dimensional shallow cutting.

\begin{figure}[tb]
  \centering
  \begin{minipage}{.5\textwidth}
    \centering
    \includegraphics[width=.4\linewidth,page=1]{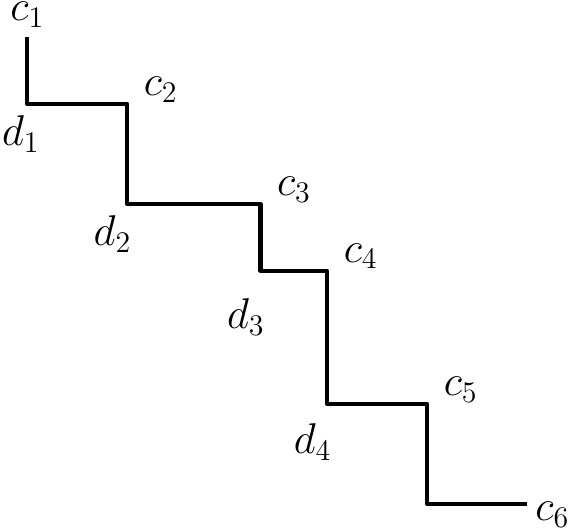}
\end{minipage}%
\begin{minipage}{.5\textwidth}
  \centering
  \includegraphics[width=.4\linewidth,page=3]{Figures/staircase}
\end{minipage}  
  \caption{Left: Inner and outer corners of a staircase (2d shallow cutting). Right: Patching of a staircase. Patch and new corners are shown with dashed polyline. The old corner $c_j$ dominates at least $7k$ points, but old corners $c_{i+1}$, $\ldots$, $c_{j-1}$ dominate less than $7k$ points.}
  \label{fig:staircase}
\end{figure}

\paragraph{Our Implementation}
The crucial difference between our result and  the previous work is that we reduce the number of updates on the two-dimensional shallow cutting. An update happens only when we have to remove an outer corner and patch the staircase. To this aim, we associate with each outer corner $c_j$ the $z$-position of the sweep-plane when it should  be removed (provided that $c_j$ was not removed earlier by the patching procedure). A more detailed description of our method follows below.

We maintain a queue $\cQ$ of corner events. $\cQ$ supports insertions, deletions, and find-max operations.  Every time when an inner corner $d_i=(x,y)$ is added to the staircase, we compute the highest value  $v$ such that the point $(x,y,v)$ dominates less than $k$ points of $P$; we associate an event  $(x,y,v)$ with $d_i$ and add it  to $\cQ$.  Obviously, $d_i$ should be removed from the staircase when the $z$-coordinate of the sweep-plane is equal to $v$ (unless $d_i$ was removed earlier by the patching procedure). We will use the notation $\fin(d_i)=v$.  We also maintain the list $\cL$ of all outer and inner corners sorted by $x$-coordinate. For each corner $c_i$ in the list we keep a pointer to the associated event in $\cQ$ (as well as a reverse pointer from an event in $\cQ$ to the corner in $\cL$).

Our algorithm also requires a data structure that supports range counting and range selection queries in 3d. A 3d orthogonal range counting query on a set $P$ computes the number of points in $q\cap P$ for a three-dimensional axis-parallel rectangle $q'=\prod_{i=1}^3 [a_i,b_i]$. 
A three-dimensional $z$-selection query $(a,b,k)$ asks for the highest value $v$  such that the point $p=(a,b,v)$ dominates exactly $k$ points in $P$\footnote{In this paper we consider a special case of selection queries, the dominance selection queries. For simplicity, this special case will be denoted selection queries.}. A three-dimensional $y$-selection query $(a,b,k)$ asks for the highest value $v$  such that the point  $p=(a,v,b)$ dominates exactly $k$ points in $P$.
A three-dimensional $x$-selection query $(a,b,k)$ asks for the highest value $v$ such that  such that the point  $p=(v,a,b)$ dominates exactly $k$ points in $P$.

The complete algorithm works as follows. We find the largest $z$-coordinate of the sweep-plane, such that at least one inner corner of the staircase dominates less than  $k$ points. This value can be found by extracting an event with the highest $z$-coordinate from the queue $\cQ$ (if there are several events with the same $z$-coordinate, we select the event with the smallest $x$-coordinate). Then we patch the corresponding corner $d_i$ as described above. Let $c_i$ denote the outer corner that precedes $d_i$.  We find the  point  $c'_i$ with $y(c'_i)=y(d_{i-1})$ that dominates  exactly $10k$ points.  Finding  $c'_i$ is equivalent to answering a range selection query. When $c'_i$ is known we find the point  $d'_i$ with the highest $y$-coordinate, such that $x(c'_i)=x(d'_i)$ and $c'_i$ dominates $9k$ points. We can find $d'_i$ using a range selection query too. Next, we identify the rightmost old corner $d_j$, such that $x(d_j)\le x(d'_i)$. The corner $d_j$ can be found by binary search on $\cL$. If $d_j$ dominates more than $7k$ points If $d_j$ dominates at most $7k$ points, we set $x(d'_i)=x(d_j)$, discard $c'_i$, and connect $d'_i$ with $d_j$ by a vertical segment. This completes the patching procedure.  When $c'_i$ and $d'_i$ are computed, we update the list $\cL$. New corners  $c_i'$ and $d_i'$ are inserted into $\cL$; old outer corners dominated by $d_i'$ and preceding inner corners are removed from $\cL$. 
For every  new inner corner $d_i$ that is inserted into $\cL$, we compute $\fin(c_i)$ and insert $(x(c_i),y(c_i),\fin(c_i))$  into $\cQ$.

In summary, we can construct a $t$-level three-dimensional shallow cutting with $g=O(\frac{m}{t})$ cells by executing $O(g)$ binary searches, executing $O(g)$ queries  and updates of $\cQ$, answering $O(g)$ three-dimensional range counting queries, and answering $O(g)$ three-dimensional range selection queries. The queue $\cQ$ can be implemented as a binary search tree, so that updates and queries can be supported in $O(\log m)$ time. Our algorithm is summarized in the following lemma.  
\begin{lemma}
  \label{lemma:construct3d}
  Suppose that we are given a data structure that supports three-dimensional orthogonal range counting queries and two-dimensional range selection queries on a set $P$ in times $f_{\mathtt{count}}(m)$ and $f_{select}(m)$ respectively, where $m$ is the number of points in $P$.   Then a  $t$-shallow cutting for a set $P$ can be constructed in $O(g( \log m + f_{\tcount}(m)+f_{select}(m))$ time, where $m$ is the number of points in $P$ and $g=O(\frac{m}{t})$ is the number of cells in the shallow cutting. 
\end{lemma}

\paragraph{Construction of  Shallow Cuttings in Trees $T_i$.}
We construct shallow cuttings for all sets $P(v,i,j)$ using the method from Lemma~\ref{lemma:construct3d}. In order to apply this method, we must answer three-dimensional orthogonal range counting and range selection queries. 

We construct a data structure that supports four-dimensional orthogonal range counting queries on the input set $P$. A three-dimensional orthogonal range counting query $q'$ on $P(v,i,j)$ is equivalent to a four-dimensional orthogonal range counting  query $[l,r]\times q'$ on $P$ where $l$ and $r$ are the smallest and the largest first coordinates of points in $P(v,i,j)$. Using range trees and fractional cascading, we can answer four-dimensional counting queries in $O(\log^3 n)$ time. We can reduce $x$-selection  to range counting via  binary search on $x$-coordinates: We store $x$-coordinates of all points from $P$ in a balanced binary tree. To answer an $x$-selection query $(a,b,k)$, we visit nodes starting with the root node. For each visited node $\nu$, we compute the number of points $k'$in $P(v,i,j)\cap (-\infty,e]\times(-\infty,a]\times (-\infty,b]$ where $e$ is the $x$-coordinate stored in $\nu$. If $k'=k$, we return $e$. If $k'<k$ we move to the right child of $\nu$, and if $k'>k$, we move to the left child of $\nu$. Since, we visit $O(\log n)$ nodes, an $x$-selection query is answered in $O(\log^4n)$ time.
We can answer $y$-selection and $z$-selection queries in the same way. 
Combining this method with Lemma~\ref{lemma:construct3d}, we can construct a shallow cutting on a set $P(v,i,j)$ in time
$O(g\log^4 n)$ where $g$ is the number of cells in the shallow cutting or in $O(\log^4n)$ time per cell. 

The data structure that supports four-dimensional range counting in $O(\log^3n)$ time can be constructed in $O(n\log^4 n)$ time.  
As explained in Lemma~\ref{lemma:size}, the total number of cells in all shallow cuttings is  $O(n\log n)$. Hence all shallow cuttings needed in our data structure can be constructed in time $O(n\log^5 n)$.

We also need to construct pointers between cells of  different shallow cuttings. To this end, we will use a separate data structure for each set $P(v,l,r)$.This structure stores the apex points of cells from  $C(v,l,r)$  and supports the following three-dimensional one-reporting queries: given a query range $q=[q_1,q_2]\times [q_x,\infty)\times [q_y,\infty)\times [q_z,\infty)$, we must report only one point from the query range. Using the dominance range reporting structure from~\cite{Afshani08} with minor modifications, we can support such queries in $O(\log n)$ time. The data structure for one-reporting queries can be constructed in $O(n\log  n)$ time.  Then, for $i=1,2,\ldots$, we visit all sets $P(v,l,r)$ associated to $T_i$. For every $P(v,l,r)$ of $T_i$ we identify five sets $P(u_j,l_j,r_j)$ such that $P(v,l,r)=\cup_{j=1}^3P(u_j,l_j,r_j)$ and $u_j$ are nodes of $T_{i+1}$. Next we consider all nested shallow cuttings of $\cC(v,l,r)$. For each cell of each nested shallow cutting, we find the cell $C_j$ of $P(u_j,l_j,r_j)$, $j=1$,$\ldots$,$5$,  that contains $C$. If  $C_j$ contains $C$, then the apex point of $C_j$ dominates the apex point of $C$. Hence we can find $C_j$ in $O(\log n)$ time by answering a query on a dominance one-reporting  data structure for  $C(u_j,l_r,r_j)$. Thus we can construct all necessary pointers between cells in $O(\log n)$ time per pointer.  Hence all pointers between cells can be computed in $O(n\log^2 n)$ time. Other components of data structure for restricted 5-sided queries (slow data structures and conflict lists of cells associated with nodes of $T_{\log_{\beta}\log n}$) can be also constructed in $O(n\log^4 n)$ time.


Our main result is the following theorem. 
\begin{theorem}
  \label{theor:construct}
  Data structure for restricted 5-sided 4d queries, described in Theorem~\ref{thm:rest-5-sided},  can be constructed in $O(n\log^5 n)$ time. 
\end{theorem}

In order to construct the data structure for (unrestricted) 5-sided queries described in Theorem~\ref{theor:5sid}, we (1) construct the data structure for restricted queries and (2) recursively construct data structures for $n^{1/3}$ sets of $n^{2/3}$ points. Hence the construction time $C(n)$ satisfies the recurrence $C(n)=O(n\log^5n)+ n^{1/3}\cdot C(n^{2/3})$.
Let $c(n)=\frac{C(n)}{n}$. Then $c(n)=O(\log^5 n)+ c(n^{2/3})$. The latter recurrence resolves to $c(n)=O(\log^5 n)$.
Hence $C(n)=O(n\log^5 n)$ and the data structure of Theorem~\ref{theor:5sid} can be constructed in $O(n\log^5 n)$ time. 

Data structures for the general case of orthogonal range reporting consist of multiple instances of structures from Theorem~\ref{theor:5sid}.
Hence, applying Theorem~\ref{theor:construct}, we can construct the data structure for $d$-dimensional orthogonal range reporting queries described in Theorem~\ref{theor:multidim}. The construction time is $O(n\log^{4+d}n)$.


\section{Proof of Lemma~\ref{lem:cover-points} }
\label{sec:app-proof}
We start with the following observation. 
\begin{observation}\label{lem:expansion}
For any node $v \in T_i$, let $P(v) \subseteq P$ be the points in its subtree.
Let $\Delta$ be the depth of node $v$ in $T_i$ (the root node is at depth~$0$).
Then there exists a node $v' \in T_{i+1}$ such that 
(a) $P(v')=P(v)$, and (b) $v'$ is at depth  $\beta\cdot\Delta$.
\end{observation}
\begin{proof}
Since the fanout of $T_i$ is $f_i=n^{1/\beta^i}$, we claim that
\begin{equation}\label{eqn:delta}
\frac{n}{f_i^{\Delta}}=\frac{n}{n^{\Delta/\beta^i}}=|P(v)|.
\end{equation}
 
Let $v'$ be the node in $T_{i+1}$ such that $P(v')=P(v)$. 
Our construction of $T_{i+1}$ ensures that such a node $v'$ exists.
Let $\Delta'$ be the depth of node $v'$ in $T_{i+1}$.
Since the fanout of $T_{i+1}$ is $f_{i+1}=n^{1/\beta^{i+1}}$, we claim that
\begin{equation}\label{eqn:delta-prime}
\frac{n}{f_{i+1}^{\Delta'}}=\frac{n}{n^{\Delta'/\beta^{i+1}}}=|P(v)|.
\end{equation}
Combining Equations~\ref{eqn:delta} and \ref{eqn:delta-prime}, we observe that $\Delta'=\beta\cdot \Delta$.
\end{proof}

Now we are ready to prove Lemma~\ref{lem:cover-points}. 
\begin{proof}
(Covering bounded ranges) In $T_i$, let $\Delta$ be the depth of node $v$. Via Observation~\ref{lem:expansion}, there exists a 
node $v'$ in $T_{i+1}$ such that $P(v')=P(v)$. 
Again via Observation~\ref{lem:expansion}, for any $P(v, \ell, r)$, there exists two 
  nodes $u'_{\ell}$ and $u'_r$ 
 at depth $\beta \cdot (\Delta+1)$ such that 
 $P(u'_{\ell})=P(u_{\ell})$ and $P(u'_r)=P(u_r)$, respectively.
 See Figure~\ref{fig:expansion}.
Let $\Pi_r$ (resp., $\Pi_{\ell}$) be the path 
from $v'$ to $u'_r$ (resp., $u'_{\ell}$). Then  assign $v_1, v_2,\ldots, v_{\beta}$ to
be the $\beta$ nodes on $\Pi_r$ (excluding $u_r'$).
Next, assign $v_{\beta+1}, v_{\beta+2},\ldots, v_{2\beta-1}$ to be the $(\beta-1)$ nodes on 
$\Pi_{\ell}$ (excluding $v'$ and $u'_{\ell}$). 

For all $1\leq j\leq (2\beta-1)$,  let the child nodes of $v_j$
be $w_1,w_2,\ldots,w_{f'}$. Then we
pick $P(v_j,\ell_j,r_j)$ if 
 $w_{\ell_j}$ is the leftmost (resp., $w_{r_j}$ is the rightmost) 
child of $v_j$ such that $P(w_{\ell_j}) \subseteq P(v, \ell, r)$ 
(resp., $P(w_{r_j}) \subseteq P(v, \ell, r)$).
Then, the reader can observe that $P(v, \ell, r)= \bigcup_{j\in [2\beta-1]} P(v_j,\ell_j,r_j)$.

Let $v_1$ be the node common to both $\Pi_{\ell}$ and $\Pi_r$.
For any node $v_j \in \Pi_r \setminus \{v_1\}$ (resp., $v_j \in \Pi_{\ell}
\setminus \{v_1\}$), 
its child node $w_{1}$ (resp., $w_{f'}$) will always satisfy 
$P(w_1) \subseteq P(v,\ell,r)$ (resp., $P(w_{f'}) \subseteq P(v,\ell,r)$). 
Hence, for all $v\in \Pi_r \setminus \{v_1\}$, we have $\ell_j=1$, and 
for all $v\in \Pi_{\ell} \setminus \{v_1\}$, we have $r_j=f'$.

(Covering prefix ranges) In $T_{i+1}$, let $\Pi_{r}$ be the path from $v'$ to $u_r'$ 
(excluding $u_r'$). Then assign $v_1, v_2,\ldots, v_{\beta}$ to be the $\beta$ nodes 
on $\Pi_r$. As before, for all $1\leq j\leq \beta$,  let the child nodes of $v_j$
be $w_1,w_2,\ldots,w_{f'}$. Then we
pick $P(v_j,1,r_j)$ if $w_{r_j}$ is the rightmost
child of $v_j$ such that resp., $P(w_{r_j}) \subseteq P(v, 1, r)$.
Then,  $P(v, 1, r)= \bigcup_{j\in [\beta]} P(v_j,1,r_j)$.
\end{proof}

\section{Conclusion and future work}
Chazelle~\cite{c90} showed that for orthogonal 
range reporting in $d$-dimensions in the 
pointer machine model, a query time 
of $\log^{O(1)}n + O(k)$ can only be achieved 
by using $\Omega\left(n\left(\frac{\log n}{\log\log n} \right)^{d-1} \right)$
space. Our data structure for 
 4d orthogonal range reporting uses $O(n\log^4n)$
 space. The immediate open problem is to close 
 the logarithmic
 factor gap in the space  
 between our upper bound and the lower bound
 of Chazelle~\cite{c90}.

 The query time 
of our data structure is $O(\log n\cdot\log\log n + k)$. An open question
is to prove that the optimal  query time of $O(\log n + k)$
can be achieved while using 
$n\log^{O(1)}n$ space, or to prove that such a 
result is not possible. The complexity of orthogonal range reporting has also been studied in the RAM model~\cite{Chazelle86}. It is interesting to compare the current state of the art in the PM model with  data structures in the RAM model. In the latter model, the data structure that achieves optimal $O(\log n/\log\log n)$ query time was recently obtained by Nekrich~\cite{Nekrich21soda}.

The ``dual'' version of orthogonal range reporting
is the {\em orthogonal rectangle stabbing} problem, 
where the input is a set of axis-parallel boxes in
a $d$-dimensional space and the query is a point;  the goal is to efficiently report all the 
boxes containing the query point. In 2d,
Chazelle~\cite{Chazelle86} presented an optimal pointer 
machine data structure which uses linear-space 
and answers a query in $O(\log n + k)$ time.
Afshani et al.~\cite{AfshaniAL12} proved that any data structure  which uses linear-space will take $\Omega(\log^{d-1}+ k)$ time to answer a query.
This lower bound is almost matched in 3d by 
the data structure of Rahul~\cite{r15}.
Currently, 
the best known data structure in 4d uses 
$O(n\log n \log^{*}n)$ space and answers a 
query in $O(\log^3n\log\log n +k)$ time~\cite{r15}.
Can the 4d data structures presented in this 
paper be useful in resolving the orthogonal 
rectangle stabbing problem in 4d?

\ifarxiv
\bibliographystyle{plain}
\bibliography{4d-dom-pm}
\fi

\ifsoda

\fi
\appendix

\end{document}